\documentclass[conference]{IEEEtran}
\usepackage[boxruled]{algorithm2e}
\usepackage{cite}
\usepackage{graphicx}
\usepackage{psfrag}
\usepackage{subfigure}
\usepackage{url}
\usepackage{amsmath}
\usepackage{array}
\usepackage{amssymb}
\usepackage{amsfonts}
\usepackage{graphicx}
\usepackage{epstopdf}

\newtheorem{lemma}{Lemma}

\newtheorem{remark}{Remark}
\newtheorem{definition}{Definition}

\newtheorem{example}{Example}

\title{Physical Layer Network Coding for Two-Way Relaying with QAM and Latin Squares}
\begin{document}

\author{
\authorblockN{Vishnu Namboodiri}
\authorblockA{Dept. of ECE, Indian Institute of Science \\
Bangalore 560012, India\\
Email: vishnukk@ece.iisc.ernet.in
}
\and
\authorblockN{B. Sundar Rajan}
\authorblockA{Dept. of ECE, Indian Institute of Science, \\Bangalore 560012, India\\
Email: bsrajan@ece.iisc.ernet.in
}
}

\maketitle
\thispagestyle{empty}	
%%%%%%%%
%%%%%%%%%%%%%%%%%%%%%%%%%%%%%%%%%%%%%%%%%%%%%%%%%%%%%%%%%%%%%%%%%%%%%%%%%%%%%%%%%%%%%
\begin{abstract}
The design of modulation schemes for the physical layer network-coded two way relaying scenario has been extensively studied recently with the protocol which employs two phases: Multiple access (MA) Phase and Broadcast (BC) Phase. It was observed by Koike-Akino et al. that adaptively changing the network coding map used at the relay according to the channel conditions greatly reduces the impact of multiple access interference which occurs at the relay during the MA Phase and all these network coding maps should satisfy a requirement called the {\it exclusive law}. In \cite{NVR} it is shown that every network coding map that satisfies the exclusive law is representable by a Latin Square and conversely, and this relationship can be used to get the network coding maps satisfying the exclusive law. But, only the  scenario in which the end nodes use $M$-PSK signal sets (where $M$ is of the form $2^\lambda$, $\lambda$ being any positive integer) is extensively studied in \cite{NVR}. In this paper, we address the case in which the end nodes use $M$-QAM signal sets (where $M$ is of the form $2^{2\lambda}$, $\lambda$ being any positive integer). In a fading scenario, for certain channel conditions $\gamma e^{j \theta}$, termed singular fade states, the MA phase performance is greatly reduced. We show that the square QAM signal sets give lesser number of singular fade states compared to PSK signal sets. Because of this, the complexity at the relay is enormously reduced. Moreover, lesser number of overhead bits are required in the BC phase. 
The fade state $\gamma e^{j \theta}=1$ is singular for all constellations of arbitrary size including PSK and QAM. For arbitrary PSK constellation it is well known that the Latin Square obtained by bit-wise XOR mapping removes this singularity. We show that XOR mapping fails to remove this singularity for QAM of size more greater than 4 and show that a doubly block circulant Latin Square removes this singularity. Simulation results are presented to show the superiority of QAM over PSK. 
\end{abstract}

%%%%%%%%%%%%%%%%%%%%%%%%%%%%%%%%%%

\section{Preliminaries and Background}
We consider the two-way wireless relaying scenario shown in Fig.\ref{relay_channel}, where bi-directional data transfer takes place between the nodes A and B with the help of the relay R. It is assumed that all the three nodes operate in half-duplex mode, i.e., they cannot transmit and receive simultaneously in the same frequency band. The relaying protocol consists of the following two phases: the \textit{multiple access} (MA) phase, during which A and B simultaneously transmit to R using identical square $M$-QAM signal sets and the \textit{broadcast} (BC) phase during which R transmits to A and B using possibly with another square $M$-QAM or constellations of size more than $M$. Network coding is employed at R in such a way that A (B) can decode the message of B (A), given that A (B) knows its own message. 
%%%%%%%%%%%%%%%%%%%%%%%%%%%%%%%%%%
\begin{figure}[htbp]
\centering
\subfigure[MA Phase]{
\includegraphics[totalheight=1in,width=2in]{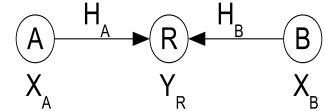}
\label{fig:phase1}
}

\subfigure[BC Phase]{
\includegraphics[totalheight=1in,width=2in]{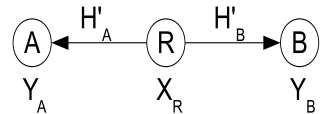}
\label{fig:phase2}
}
\caption{The Two Way Relay Channel}
\label{relay_channel}
\end{figure}
\vspace{-0.1 cm}
%%%%%%%%%%%%%%%%%%%%%%%%%%%%%%%%
%%%%%%%%%%%%%%%%%%%%%%%%%%%%%% 
\subsection{Background}
 
 The concept of physical layer network coding has attracted a lot of attention in recent times. The idea of physical layer network coding for the two way relay channel was first introduced in \cite{ZLL}, where the multiple access interference occurring at the relay was exploited so that the communication between the end nodes can be done using a two stage protocol. Information theoretic studies for the physical layer network coding scenario were reported in \cite{KMT},\cite{PoY}. The design principles governing the choice of modulation schemes to be used at the nodes for uncoded transmission were studied in \cite{APT1}. An extension for the case when the nodes use convolutional codes was done in \cite{APT2}. A multi-level coding scheme for the two-way relaying scenario was proposed in \cite{HeN}.

It was observed in \cite{APT1} that for uncoded transmission, the network coding map used at the relay needs to be changed adaptively according to the channel fade coefficients, in order to minimize the impact of the multiple access interference. The proposed Latin Square scheme was studied in \cite{NVR}, \cite{VNR} by considering a two way relaying using $M$-PSK signal sets at the end nodes. %In \cite{SR} three way relaying using $M$-PSK was studied using Latin Cubes. 
In \cite{APT1} analysis of 16-QAM is done under the assumption that precoding is done at the end nodes. We address  the situation where no precoding assumption is made and to the best of our knowledge no work has been reported for such a scenario with general $M$-QAM modulation. 

\subsection{Signal Model}
\subsubsection*{Multiple Access (MA) Phase}

Let $\mathcal{S}$ denote the square $M$-QAM constellation used at A and B, where $M=2^{2\lambda}$, $\lambda$ being a positive integer. Assume that A (B) wants to transmit an $2\lambda$-bit binary tuple to B (A). Let $\mu:\mathcal{S}  \rightarrow \mathbb{F}_{2^{2\lambda}}$ denote the mapping from complex symbols to bits used at A and B. Let $\mu(x_A)= s_A$, $\mu(x_B)=s_B \in \mathcal{S}$ denote the complex symbols transmitted by A and B respectively, where $s_A,s_B \in \mathbb{F}_{2^{2\lambda}}$. The received signal at $R$ is given by,
\begin{align}
\nonumber
Y_R=H_{A} x_A + H_{B} x_B +Z_R,
\end{align}
where $H_A$ and $H_B$ are the fading coefficients associated with the A-R and B-R links respectively. The additive noise $Z_R$ is assumed to be $\mathcal{CN}(0,\sigma^2)$, where $\mathcal{CN}(0,\sigma^2)$ denotes the circularly symmetric complex Gaussian random variable with variance $\sigma ^2$. We assume a block fading scenario, with the ratio $ H_{B}/H_{A}$ denoted as $z=\gamma e^{j \theta}$, where $\gamma \in \mathbb{R}^+$ and $-\pi \leq \theta < \pi,$ is referred as the {\it fade state} and for simplicity, also denoted by $(\gamma, \theta).$    
 
 Let $\mathcal{S}_{R}(\gamma,\theta)$ denote the effective constellation at the relay during the MA Phase, i.e., 
\begin{align} 
\nonumber
 \mathcal{S}_{R}(\gamma,\theta)=\left\lbrace x_i+\gamma e^{j \theta} x_j \vert x_i,x_j \in \mathcal{S}\right \rbrace,
 \end{align}
and $d_{min}(\gamma e^{j\theta})$ denote the minimum distance between the points in $\mathcal{S}_{R}(\gamma,\theta)$, i.e.,

{\footnotesize
\begin{align}
%\nonumber
\label{eqn_dmin} 
d_{min}(\gamma e^{j\theta})=\hspace{-0.5 cm}\min_{\substack {{(x_A,x_B),(x'_A,x'_B)}{ \in \mathcal{S}^2 } \\ {(x_A,x_B) \neq (x'_A,x'_B)}}}\hspace{-0.5 cm}\vert \left(x_A-x'_A\right)+\gamma e^{j \theta} \left(x_B-x'_B\right)\vert.
\end{align}
}
 
 From \eqref{eqn_dmin}, it is clear that there exists values of $\gamma e^{j \theta}$ for which $d_{min}(\gamma e^{j\theta})=0$. Let $\mathcal{H}=\lbrace \gamma e^{j\theta} \in \mathbb{C} \vert d_{min}(\gamma,\theta)=0 \rbrace$. The elements of $\mathcal{H}$ are said to be {\it the singular fade states}. Singular fade states can also be defined as

\begin{definition}
 A fade state $\gamma e^{j \theta}$ is said to be a singular fade state, if the cardinality of the signal set $\mathcal{S}_{R}(\gamma, \theta)$ is less than $M^2$.
\end{definition}
 
 For example, consider the case when symmetric 4-QAM signal set used at the nodes A and B, i.e., $\mathcal{S}=\lbrace (\pm 1 \pm j)/\sqrt{2} \rbrace$. For $\gamma e^{j \theta}=(1+j)/2$, $d_{min}(\gamma e^{j \theta})=0$, since,
\begin{align*} 
 \left\vert \left( \dfrac{1+j}{\sqrt{2}}-\dfrac{1-j}{\sqrt{2}} \right) + \dfrac{(1+j)}{2} \left( \dfrac{-1-j}{\sqrt{2}} - \dfrac{1+j}{\sqrt{2}} \right)\right\vert=0.
 \end{align*}
\noindent 
Alternatively, when $\gamma e^{j \theta}=(1+j)/2$, the constellation $\mathcal{S}_{R}(\gamma,\theta)$ has only 12 ($<$16) points. 
 Hence $\gamma e^{j \theta}=(1+j)/2$ is a singular fade state for the case when 4-QAM signal set is used at A and B.
 Let $(\hat{x}_A,\hat{x}_B) \in \mathcal{S}^2$ denote the Maximum Likelihood (ML) estimate of $({x}_A,{x}_B)$ at R based on the received complex number $Y_{R}$, i.e.,
 \begin{align}
 (\hat{x}_A,\hat{x}_B)=\arg\min_{({x}'_A,{x}'_B) \in \mathcal{S}^2} \vert Y_R-H_{A}{x}'_A-H_{B}{x}'_B\vert.
 \end{align}
%%%%%%%%%%%%%%%%%%%%%%%%%%%%%%%%%%%%%%%%%%%%%
\subsubsection*{Broadcast (BC) Phase}

Depending on the value of $\gamma e^{j \theta}$, R chooses a map $\mathcal{M}^{\gamma,\theta}:\mathcal{S}^2 \rightarrow \mathcal{S}'$, where $\mathcal{S}'$ is the signal set (of size between $M$ and $M^2$) used by R during $BC$ phase. The elements in $\mathcal{S}^2 $ which are mapped on to the same complex number in $\mathcal{S}'$ by the map $\mathcal{M}^{\gamma,\theta}$ are said to form a cluster. Let $\lbrace \mathcal{L}_1, \mathcal{L}_2,...,\mathcal{L}_l\rbrace$ denote the set of all such clusters. The formation of clusters is called clustering, and denoted by $\mathcal{C}^{\gamma e^{j\theta}}$ to indicate that it is a function of $\gamma e^{j \theta}.$ The received signals at A and B during the BC phase are respectively given by,
\begin{align}
Y_A=H'_{A} X_R + Z_A,\;Y_B=H'_{B} X_R + Z_B,
\end{align}
where $X_R=\mathcal{M}^{\gamma,\theta}(\hat{x}_A,\hat{x}_B) \in \mathcal{S'}$ is the complex number transmitted by R. The fading coefficients corresponding to the R-A and R-B links are denoted by $H'_{A}$ and $H'_{B}$ respectively and the additive noises $Z_A$ and $Z_B$ are $\mathcal{CN}(0,\sigma ^2$).

In order to ensure that A (B) is able to decode B's (A's) message, the clustering $\mathcal{C}$ should satisfy the exclusive law \cite{APT1}, i.e.,

{\footnotesize
\begin{align}
\left.
\begin{array}{ll}
\nonumber
\mathcal{M}^{\gamma,\theta}(x_A,x_B) \neq \mathcal{M}^{\gamma,\theta}(x'_A,x_B), \; \mathrm{for} \;x_A \neq x'_A \; \mathrm{,} \; \forall x_B \in  \mathcal{S},\\
\nonumber
\mathcal{M}^{\gamma,\theta}(x_A,x_B) \neq \mathcal{M}^{\gamma,\theta}(x_A,x'_B), \; \mathrm{for} \;x_B \neq x'_B \; \mathrm{,} \;\forall x_A \in \mathcal{S}.
\end {array}
\right\} \\
\label{ex_law}
\end{align}
\vspace{-.3 cm}
}

\begin{definition}
The cluster distance between a pair of clusters $\mathcal{L}_i$ and $\mathcal{L}_j$ is the minimum among all the distances calculated between the points $x_A+\gamma e^{j\theta} x_B ,x'_A+\gamma e^{j\theta} x'_B \in \mathcal{S}_R(\gamma,\theta)$ where $(x_A,x_B) \in \mathcal{L}_i$ and $(x'_A,x'_B) \in \mathcal{L}_j.$ The \textit{minimum cluster distance} of the clustering $\mathcal{C}$ is the minimum among all the cluster distances, i.e.,

{\footnotesize
\begin{align}
\nonumber
d_{min}^{\mathcal{C}}(\gamma e^{j \theta})=\hspace{-0.8 cm}\min_{\substack {{(x_A,x_B),(x'_A,x'_B)}\\{ \in \mathcal{S}^2,} \\ {\mathcal{M}^{\gamma,\theta}(x_A,x_B) \neq \mathcal{M}^{\gamma,\theta}(x'_A,x'_B)}}}\hspace{-0.8 cm}\vert \left( x_A-x'_A\right)+\gamma e^{j \theta} \left(x_B-x'_B\right)\vert.
\end{align}
}

\end{definition}

The minimum cluster distance determines the performance during the MA phase of relaying. The performance during the BC phase is determined by the minimum distance of the signal set $\mathcal{S}'$. For values of $\gamma e^{j \theta}$ in the neighborhood of the singular fade states, the value of $d_{min}(\gamma e^{j\theta})$ is greatly reduced, a phenomenon referred as {\it distance shortening}. To avoid distance shortening, for each singular fade state, a clustering needs to be chosen such that the minimum cluster distance at the singular fade state is non-zero and is also maximized.  

A clustering $\mathcal{C}$ is said to remove a singular fade state $ h \in \mathcal{H}$, if $d_{min}^{\mathcal{C}}(h)>0$. 
For a singular fade state $h \in \mathcal{H}$, let $\mathcal{C}_{\lbrace h\rbrace}$ denote a clustering which removes the singular fade state $h$ (if there are multiple clusterings which remove the same singular fade state $h$, consider a clustering which maximizes the minimum cluster distance). Let $\mathcal{C}_{\mathcal{H}}=\left\lbrace \mathcal{C}_{\lbrace h\rbrace} : h \in \mathcal{H} \right\rbrace$ denote the set of all such clusterings. Let $d_{min}({\mathcal{C}^{\lbrace h\rbrace}},\gamma',\theta')$ be defined as,

{\footnotesize
\begin{align}
\nonumber
d_{min}({\mathcal{C}^{\lbrace h\rbrace}},\gamma',\theta')=\hspace{-0.8 cm}\min_{\substack {{(x_A,x_B),(x'_A,x'_B) \in \mathcal{S}^2,} \\ {\mathcal{M}^{\lbrace h\rbrace}(x_A,x_B) \neq \mathcal{M}^{\lbrace h \rbrace}(x'_A,x'_B)}}}\hspace{-0.8 cm}\vert \left( x_A-x'_A\right)+\gamma' e^{j \theta'} \left(x_B-x'_B\right)\vert.
\end{align}
}

The quantity $d_{min}({\mathcal{C}^{\lbrace h\rbrace}},\gamma,'\theta')$ is referred to as  the minimum cluster distance of the clustering $\mathcal{C}^{\lbrace h\rbrace}$ evaluated at $\gamma' e^{j\theta'}.$

In practice, the channel fade state need not be a singular fade state. In such a scenario, among all the clusterings which remove the singular fade states, the one which maximizes the minimum cluster distance is chosen. In other words, for $\gamma' e^{j \theta'} \notin \mathcal{H}$, the clustering $\mathcal{C}^{\gamma',\theta'}$ is chosen to be $\mathcal{C}^{\lbrace h\rbrace}$, which satisfies $d_{min}({\mathcal{C}^{\lbrace h\rbrace}},\gamma',\theta') \geq d_{min}({\mathcal{C}^{\lbrace h' \rbrace}},\gamma',\theta'), \forall h \neq h' \in \mathcal{H}$. Since the clusterings which remove the singular fade states are known to all the three nodes and are finite in number, the clustering used for a particular realization of the fade state can be indicated by R to A and B using overhead bits.

%For $\gamma e^{j \theta} \notin \mathcal{H}$, the clustering $\mathcal{C}$ is chosen to be $\mathcal{C}_{\lbrace h\rbrace}$, which satisfies $d_{min}^{\mathcal{C}_{\lbrace h\rbrace}}(\gamma e^{j \theta}) \geq d_{min}^{\mathcal{C}_{\lbrace h' \rbrace}}(\gamma e^{j \theta}), \forall h \neq h' \in \mathcal{H}$.
%%%%%%%%%%%%%%%%
\begin{example}
In the case of BPSK, if channel condition is $\gamma=1$ and $\theta=0$ the distance between the pairs $(0,1)(1,0)$ is zero as in Fig.\ref{fig:BPSK}(a).The following clustering remove this singular fade state.
$$\{\{(0,1)(1,0)\},\{(1,1)(0,0)\}\}$$
The minimum cluster distance is non zero in this clustering.
%%%%%%%%%%%%%%%% 
 \begin{figure}[h]
\centering
\vspace{-.45 cm}
\includegraphics[totalheight=2.5in,width=2.5in]{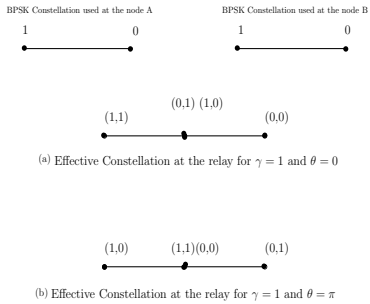}
\vspace{-1 cm}
\caption{Effective Constellation at the relay for singular fade states, when the end nodes use BPSK constellation.}     
\label{fig:BPSK}        
\end{figure}
\end{example}
%%%%%%%%%%%%%%%%%%%%%%%%%%%%%%%%%%%%%%%%%%%%%%%%%%

To remove the distance shortening effect a procedure is given in \cite{APT1} when the nodes A and B use QPSK signal set. The procedure suggested in \cite{APT1} to obtain the channel quantization and the clusterings, was using a computer algorithm, which involved varying the fade state values over the entire complex plane, i.e., $0 \leq  \gamma < \infty$, $0 \leq \theta < 2\pi$ in small discrete steps and finding the clustering for each value of channel realization. But such an approach have many issues. In \cite{APT1}, it is claimed that the clustering used by the relay is indicated to A and B by using overhead bits. However, the procedure suggested in \cite{APT1} to obtain the set of all clusterings, was using a computer search algorithm (called Closest Neighbour Clustering (CNC) algorithm), which involved varying the fade state values over the entire complex plane, i.e., $0 \leq  \gamma < \infty$, $0 \leq \theta < 2\pi$ and finding the clustering for each value of channel realization as discussed in previous sections. The total number of network codes which would result is known only after the algorithm is run for all possible realizations $\gamma e^{j \theta}$ which is uncountably infinite and hence the number of overhead bits required is not known beforehand. Moreover, performing such an exhaustive search is extremely difficult in practice, especially when the cardinality of the signal set $M$ is large.

The implementation complexity of CNC suggested in \cite{APT1} is extremely high: It appears that, for each realization of the singular fade state, the CNC algorithm of \cite{APT1} needs to be run at R to find the clustering. 

In the CNC algorithm suggested in \cite{APT1}, the network coding map is obtained by considering the entire distance profile. The disadvantages of such an approach are two-fold. 
\begin{itemize}
\item
Considering the entire distance profile, instead of the minimum cluster distance alone which contributes dominantly to the error probability, results in an extremely large number of network coding maps. For example, for 16-QAM signal set, the CNC algorithm results in more than 18,000 maps \cite{APT1}.
\item
 The CNC algorithm tries to optimize the entire distance profile, even after clustering signal points which contribute the minimum distance. As a result, for several channel conditions, the number of clusters in the clustering obtained is greater than the number of clusters in the clustering obtained by taking the minimum distance alone into consideration. This results in a degradation in performance during the BC phase, since the relay uses a signal set with cardinality equal to the number of clusters. For example, for 16-QAM signal set, the relay has to use signal sets of cardinality 16 to 29 \cite{APT1}. 
 \end{itemize}

In \cite{APT1}, to overcome the two problems mentioned above, another algorithm is proposed, in which for a given $\gamma e^{j \theta}$, an exhaustive search is performed among all the network coding maps obtained using the closest-neighbour clustering algorithm and a map with minimum number of clusters is chosen. The difficulties associated with the implementation of the CNC algorithm carry over to the implementation of this algorithm as well.
%%%%%%%%%%%%%%%%%%%%%%%%%%%%%%%%%%%%%%%%%%%%%%%%

The contributions and organization of the paper are as follows:
\begin{itemize}
\item A procedure to obtain the number of singular fade states for PAM and QAM signal sets is presented.  
\item It is shown that for the same number of signal points $M$, the number of singular fade states for square $M$-QAM is lesser in comparison with the number of singular fade states for $M$-PSK. The advantages of this result are two fold - QAM offers better distance performance in MA Phase and QAM requires lesser number of Latin squares (i.e., a reduction in number of overhead bits).
\item To remove the singular fade state $(\gamma=1, \theta=0)$ for $\sqrt M$-PAM, a Latin Square is constructed. It is shown that the bit-wise XOR mapping cannot remove the singular fade state $(\gamma=1, \theta=0)$ for any  $M$-QAM and a different mapping is obtained to remove the singular fade state $(\gamma=1, \theta=0)$, from the Latin Square to remove the singular fade state $(\gamma=1, \theta=0)$ for $\sqrt M$-PAM.
\item By simulation it is shown that the choice of 16-QAM leads to better performance for both the Rayleigh and the Rician fading scenario, compared to 16-PSK. 
\end{itemize}

The remaining content is organized as follows:  

In Section \ref{sec2} we discuss the relationship between singular fade states and difference constellation of the signal sets used by the end nodes. We present expressions to get the number of singular fade states for PAM and square QAM signal sets in Subsections \ref{subsec_1_2} and \ref{subsec_2_2} respectively. In Subsection \ref{subsec_3_2} it is proved that the number of singular fade states for $M$-QAM is always lesser in comparison with that of $M$-PSK signal sets.  In Section \ref{sec3} the clustering for a singular fade state is obtained through completing a Latin Square and a Latin Square for removing the singular fade state $z=1$ is analytically obtained for PAM and QAM signal sets. In Section \ref{sec4} simulation results are provided to show the advantage of Latin Square scheme for QAM over XOR network coding scheme as well as Latin Square scheme for PSK signal sets under Rayleigh and Rician fading channel assumptions.

 %%%%%%%%%%%%%%%%%%%%%%%%%%%%%%%%%%%%%%%%%%%%%%%%%%%%%%%%
 \section{Singular Fade states and Difference Constellations}
 \label{sec2}
In this section we show the relationships between singular fade states and difference constellation of the signal set used by the end nodes. %Throughout in our discussion we exclude the singular fade states $z=0$ and $z=\infty$ since they are irremovable singular fade states.
The following lemma discusses the location of singular fade states in complex plane for any constellation used at end nodes. 
%%%%%
\begin{lemma}
\label{sfs}
 Let node A use a constellation $\mathcal{S}_1$ of size $M_1$ and let node B use a constellation $\mathcal{S}_{2}$ of size $M_2$. Let $x_A,x_A^{\prime} \in \mathcal{S}_1$ and $x_B,x_B^{\prime} \in \mathcal{S}_2$, then the singular fade states $z=\gamma e^{j \theta}$ are given by
 \begin{equation}
 \label{sing_expression}
 z=\gamma e^{j\theta}=\dfrac{x_A-x_A^{\prime}}{x_B^{\prime}-x_B} 
\end{equation}
where $x_A,x_A^{\prime} \in \mathcal{S}_1$ and $x_B,x_B^{\prime} \in \mathcal{S}_2$.
  \end{lemma}
%%%%%%%%%%%%%%%%%%%%%%%%%%%%%%%%%%
\begin{proof}
The pair $(x_A,x_A^{\prime})$ and $(x_B,x_B^{\prime})$ result in the same point in the effective constellation at the relay if 
 the complex numbers $x_A+\gamma e^{j \theta} x_B$ and $x_A^{\prime}+\gamma e^{j\theta}x_B^{\prime} $ are the same. The expression \eqref{sing_expression} is obtained by equating these complex numbers.
\end{proof}
%%%%%%%%%%%%%%%%%%%%%%%%%%%%%%%%%%%%%%%%%%%%%%%%%%%%%

From Lemma \ref{sfs} it can be seen that all the singular fade states in the complex plane is of the form of ratio of difference constellation points of the signal sets used by end nodes, i.e., the singular fade states are decided by the difference constellation points. Henceforth, throughout the paper,  we assume both the end nodes use same constellation, $\mathcal{S}$. Let $\Delta\mathcal{S}$ denote the difference constellation of the signal set used at the end nodes $\mathcal{S}$, i.e., $\Delta\mathcal{S}=\lbrace x_i-x'_i \vert  x_i, x'_i \in \mathcal{S}\rbrace$. For a fade state $z=\gamma e^{j\theta}$ to become a singular fade state, it has to satisfy  \eqref{sing_expression}, or in other words $z(x_B^{\prime}-x_B)=(x_A-x_A^{\prime})$, where $(x_B^{\prime}-x_B)$ and $(x_A-x_A^{\prime})$ are any point in $ \Delta\mathcal{S}$. Hence, a singular fade state can be alternatively defined as follows.
%%%%%
\begin{definition}
\label{sfs_alter_def}
A singular fade state $z$ is a mapping $\mathcal{Z}$ from  $\Delta\mathcal{S}$ to the complex plane $\mathbb{C}$ so that at least one $d_k \in \Delta\mathcal{S}$ is mapped to some $d_l \in \Delta\mathcal{S}$. The set of all singular fade states is given by $\{\mathcal{Z}:  \Delta\mathcal{S} \rightarrow \mathbb{C} \vert \hspace{.05cm} \exists \hspace{.05cm} \mathcal{Z}(d_k)= d_l\}$.
\end{definition}
%%%%

\begin{remark}
The singular fade state $z=1$ is the mapping from $\Delta\mathcal{S}$ to itself that maps every point to itself.  
\end{remark}

In the rest of this section,  we focus on the singular fade states for symmetric PAM and square QAM signal sets. 
 %%%%%%%%%%%%%%%%%%%%%%%%%%%%%%%%%%%%%%%%%%%%%%%%%%%%%%%%%%%%%%
\subsection{Singular Fade States of PAM signal sets} 
\label{subsec_1_2}
 %%%%%%%%%%%%%%%%%%%%%%%%%%%%%%%%%%%%%%%%%%%%%%%%%%%%%%%%%%%%%%
 The symmetric $\sqrt M$-PAM signal set is given by 
$$ \mathcal{S}= -(\sqrt M -1)+ 2n,  ~~~  n \in (0, \cdots ,\sqrt M -1) $$ and its difference constellation is given by $$ \Delta\mathcal{S}= -2(\sqrt M -1)+ 2n, ~~~ n \in (0, \cdots ,2(\sqrt M -1)).$$ For example the 4-PAM signal set and it's difference constellation is given in Fig.\ref{fig:pam} and Fig.\ref{fig:pamdiff} respectively. For each of the difference constellation point, the pair in the signal set which correspond to this point is also shown.
%%%%%%%%%%%%%%%%%%%%%%%%%%%%%%%%%%%%%%%%%%%
 \begin{figure}[t]
\centering
\subfigure[$\sqrt M$ PAM constellation]{
\includegraphics[totalheight=.4in,width=2in]{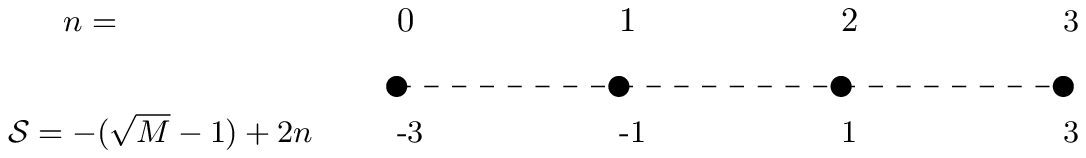}
\label{fig:pam}
}

\subfigure[Difference Constellation]{
\includegraphics[totalheight=.8in,width=3.25in]{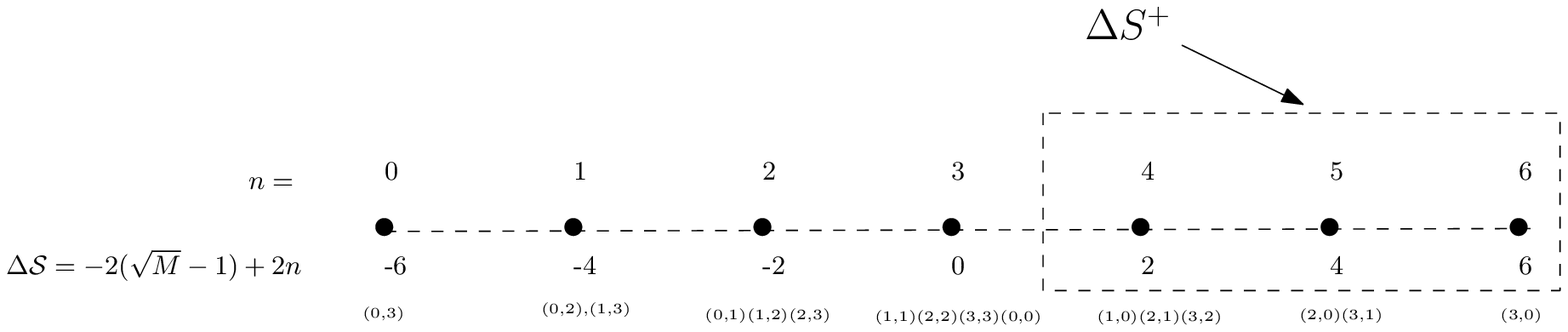}
\label{fig:pamdiff}
}

\subfigure[Singular fade states]{
\includegraphics[totalheight=1.5in,width=2.75in]{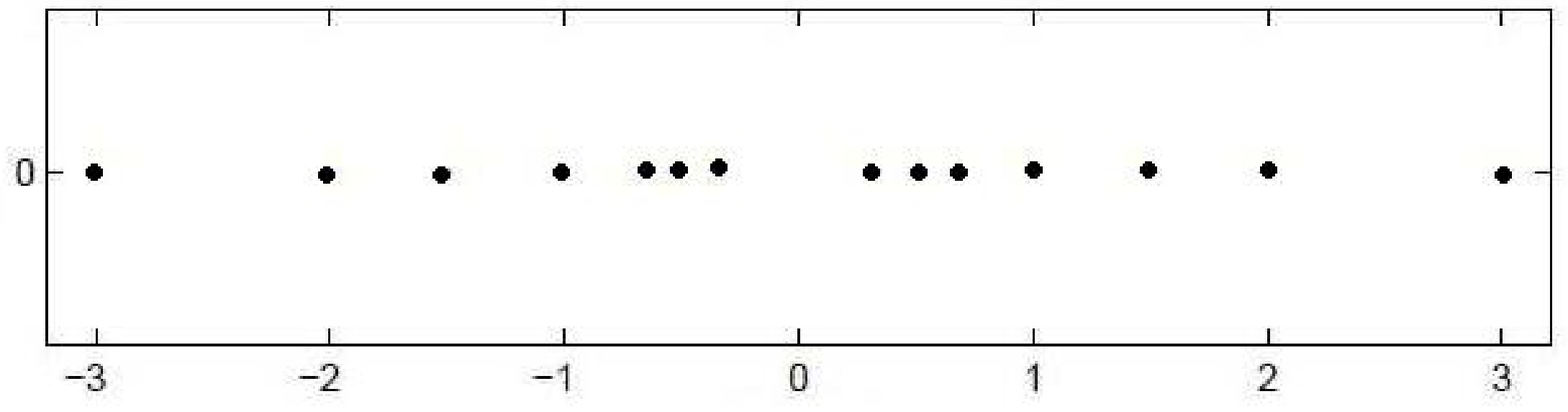}
\label{fig:pam_sing}
}
\caption{$\sqrt M$ PAM constellation, difference constellation and singular fade states for $\sqrt M=4$}
\label{fig:4pam}
\end{figure}
%%%%%%%%%%%%%%%%%%%%%%%%%%%%%%%%%%%%%%%%%%%%%%%%
We will often consider only the first quadrant of $\Delta{S}$ only, denoted as $\Delta{S}^{+},$ which for a general complex signal set is given by  $$\Delta{S}^{+} = \{\alpha : \mbox{real}(\alpha)>0, \mbox{imaginary}(\alpha) \geq 0\}.$$ %From Definition \ref{sfs_alter_def}, it can be seen that the number of distinct singular fade states is equal to the number of different mappings $\mathcal{Z}$ possible. 
The following lemma gives the number of singular fade states for PAM signal sets.
 
 %%%%%%%%%%%%%%%%%%%%%%%%%%%%%%%%%%%%%%%%%%%%%%%
 \begin{lemma}
 \label{no_sing_pam}
 The number of singular fade states, for a regular $\sqrt M$-PAM signal set, denoted by $N_{(\sqrt M-PAM)}$ is given by
 \begin{equation}
 \label{sum_euler}
 N_{(\sqrt M-PAM)}= 2 + 4\sum_{n=1}^{\sqrt M -1}n \prod_{p\vert n} \left(1-\frac{1}{p}\right)
 \end{equation}
 where $p|n$ stands for  prime number $p$ dividing $n.$
 \end{lemma} 

%%%%%%%%%%%%%%%%%%%%%%%%%%%%%%%%
\begin{proof}
There are $2(\sqrt M-1)$ non-zero signal points in the difference constellation $\Delta\mathcal{S}$ and since $ \Delta\mathcal{S}$ is symmetric about zero there are $\sqrt M-1$ signal points in $\Delta{S}^+$. All these are scaled version of nonzero elements of $\mathbb{Z}_{\sqrt M}$. %The number of singular fade states becomes a function of the number of different relatively prime pairs available in $\mathbb{Z}_{\sqrt M}$. 

The number of positive integers less than or equal to $n$ that are relatively prime to $n$ is given by Euler's totient function, 
\begin{align*}
 \psi(n)= n \prod_{p|n} \left(1-\frac{1}{p}\right)
\end{align*}
\noindent
where the product is taken over distinct prime numbers $p$ dividing $n$.  To get the total number of  relatively prime pairs in $\mathbb{Z}_{\sqrt M}$, we take the sum over all nonzero $n \in \mathbb{Z}_{\sqrt M}$  which gives $ \sum_{n=1}^{\sqrt M -1}n \prod_{p\vert n} \left(1-\frac{1}{p}\right).$ One relatively prime pair $(a,b)$ gives two singular fade states, $a/b$ and $b/a$. The multiplication factor $4$ in \eqref{sum_euler} accounts for the  negative side of the in-phase axis as well as the inverses. Finally, 2 is added to count the singular fade state $z=1$ and $z=-1$.
\end{proof}
%%%%%%%%%%%%%%%%%%%%%%%%%%%%%%%%%%
\begin{example}
Consider the case of 4-PAM ($M=16$) signal set as given in Fig.\ref{fig:4pam}. There are $2(\sqrt M-1)=6$ non-zero signal points in the difference constellation. Scaled $\Delta{S}^+$ is having $(\sqrt M-1)=3$ signal points-$\{1,2,3\}$. And there are 14 singular fade states-
\begin{align*}
\left\{1,\frac{1}{2},\frac{1}{3},\frac{2}{3},2,3,\frac{3}{2},-1,\frac{-1}{2}, \frac{-1}{3},\frac{-2}{3},-2,-3,\frac{-3}{2}\right\}.
\end{align*} 
These singular fade states are shown in Fig.\ref{fig:pam_sing}. Calculating Euler totient function $\psi(n)$ for $n=$ 1,2,3 we get 0,1,2 respectively and substituting in \eqref{sum_euler}, leads to $N_{(4-PAM)}= 2+ 4(0+1+2) =14.$
%%%%%%%%%%%%%%%% 
%  \begin{figure}[htbp]
% \centering
% \vspace{-.4 cm}
% \includegraphics[totalheight=2in,width=3.5in]{pam_sing.eps}
% \vspace{-2.6 cm}
% \caption{Singular fade states of 4-PAM signal set}     
% \label{fig:pam_sing}        
% \end{figure}
\end{example}
%%%%%%%%%%%%%%%%%%%%%%%%%%%%%%%%%%%%%%%%%%%%%%%%%%
\begin{example}
For 8-PAM signal set the singular fade states with $z>1$ are shown in Table.\ref{prifactor_pam}. For each such $z$ given in the table there exists singular fade states $-z,\frac{1}{z}$ and $-\frac{1}{z}.$ Hence, totally, there are 70 ($ 2+4(0+1+2+2+4+2+6) $) singular fade states.
%%%%%%%%%%%%%%%%%%%%%%%%%%%%%%%%%%%%%%%%%%%%%%%%%%%
%%%%%%%%%%%%%%%%%%%%%%%%%%%%%%%%%%%%%%%%%%%
 \begin{figure}[htbp]
\centering
\subfigure[$16-$QAM constellation]{
\includegraphics[totalheight=2.5in,width=2.5in]{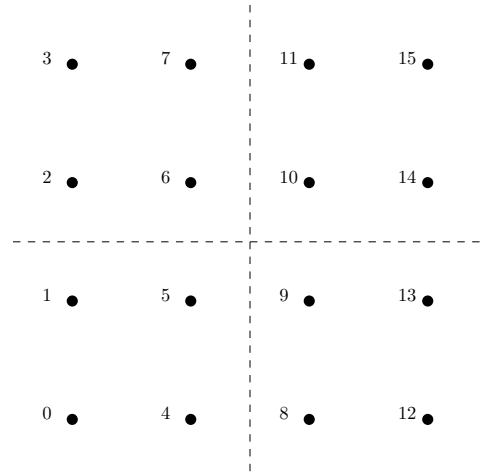}
\label{fig:16qam}
}
\subfigure[The Difference Constellation]{
\includegraphics[totalheight=2.7in,width=2.75in]{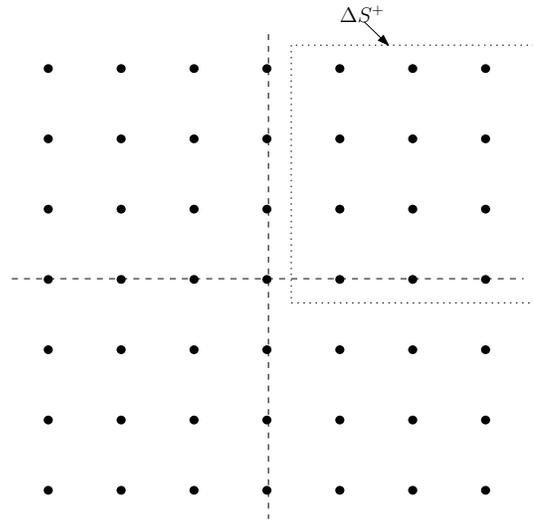}
\label{fig:16qam_diff_cons}
}
\caption{$16-$QAM constellation and its difference constellation}
\label{fig:qam}
\end{figure}
 %%%%%%%%%%%%%%%%%%%%%%%%%%%%
\begin{table}
\centering
 \caption{Singular fade states for 8-PAM}
 \label{prifactor_pam}
\begin{tabular}{|c|c|c|c|}
\hline $n$, Elements  & Relative primes & Singular fade & $\psi(n)$\\ 
 in $\Delta{S}^{+}$ & less than $n$  & states, $z >1$& \\
\hline 1 &   &   & 0\\ 
\hline 2 & 1 & 2 & 1\\
\hline 3 & 1,2 & 3,$\frac{3}{2}$ & 2\\
\hline 4 & 1,3 & 4,$\frac{4}{3}$ & 2\\
\hline 5 & 1,2,3,4 & 5,$\frac{5}{2},\frac{5}{3},\frac{5}{4}$ & 4\\
\hline 6 & 1,5 & 6,$\frac{6}{5}$ & 2\\
\hline 7 & 1,2,3,4,5,6 & 7,$\frac{7}{2},\frac{7}{3},\frac{7}{4},\frac{7}{5},\frac{7}{6}$ & 6\\
\hline 
\end{tabular}
\vspace{-.1 cm}
\end{table}
%%%%%%%%%%%%%%%%
%%%%%%%%%%%%%%%%%%%%%%%%%%%%%%%%%%%%%%%%%%%%%%%%%%
\end{example}
%%%%%%%%%%%%%%%%%%%%%%%%%%%%%%%%%%%%%%%%%%%%%%%%%%
\subsection{Singular Fade States for QAM signal sets}
\label{subsec_2_2}
We consider square $M$-QAM signal set $\mathcal{S}=\{A_{mI}+jA_{mQ}\}$ where $A_{mI}$ and $A_{mQ}$ take values from the $\sqrt M$-PAM signal set $ -(\sqrt M -1)+ 2n,  ~~~  n \in (0, \cdots ,\sqrt M -1).$ We use the mapping $\mu: \mathcal{S} \rightarrow \mathbb{Z}_m$ given by 
{\small
\begin {equation}
\label{mumap}
A_{mI}+jA_{mQ} \rightarrow \frac{1}{2}[(\sqrt M -1 +A_{mI})\sqrt M + (\sqrt M -1 +A_{mQ})]
\end{equation}
}
\noindent
for concreteness and our analysis and results hold for any map.  The difference constellation $\Delta\mathcal{S}$ of square QAM signal sets form a part of scaled integer lattice with $(2\sqrt M -1)^2$ points. The 16-QAM signal set with the above mapping and its difference constellation is shown in Fig.\ref{fig:16qam} and in Fig.\ref{fig:16qam_diff_cons}. 
%%%%%%%%%%%%%%%%%%%%%%%
\begin{figure}[htbp]
\centering
\vspace{-.5cm}
\subfigure[$4-$QAM constellation]{
\includegraphics[totalheight=1.1in,width=1.1in]{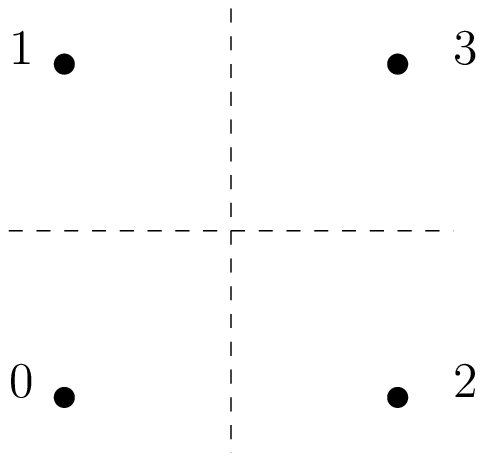}
\label{fig:4qam}
}
\subfigure[The Difference Constellation]{
\includegraphics[totalheight=2.2in,width=2.2in]{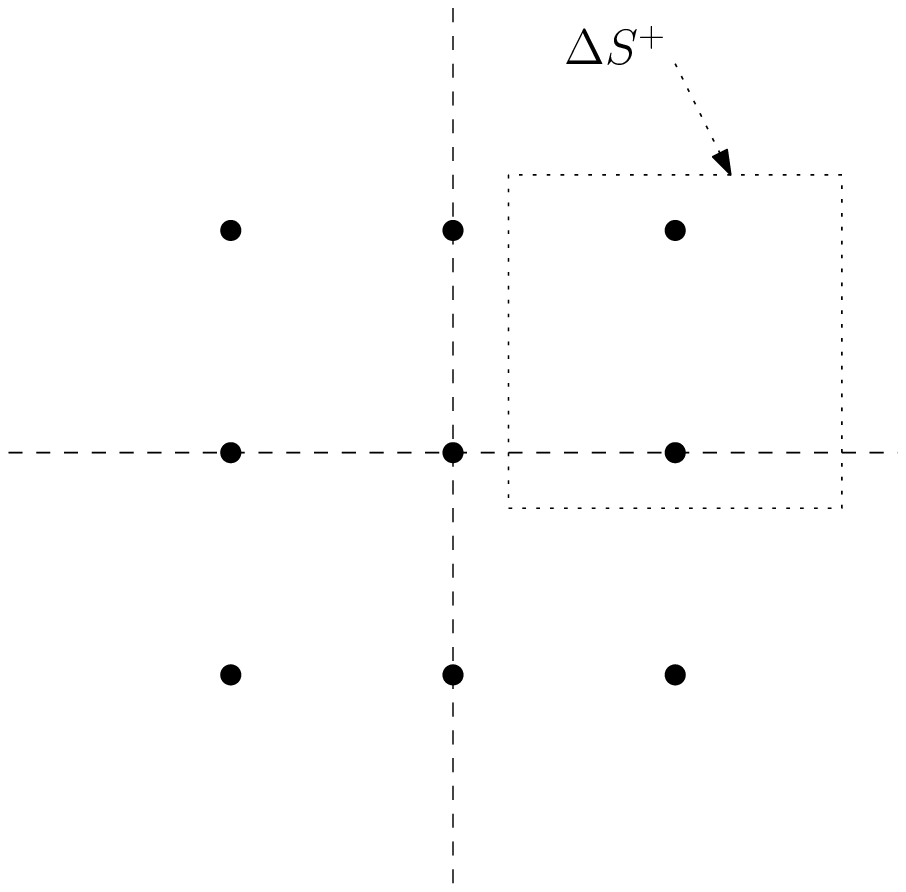}
\label{fig:4qamdiffcons}
}
\subfigure[Singular fade states]{
\includegraphics[totalheight=1.8in,width=2in]{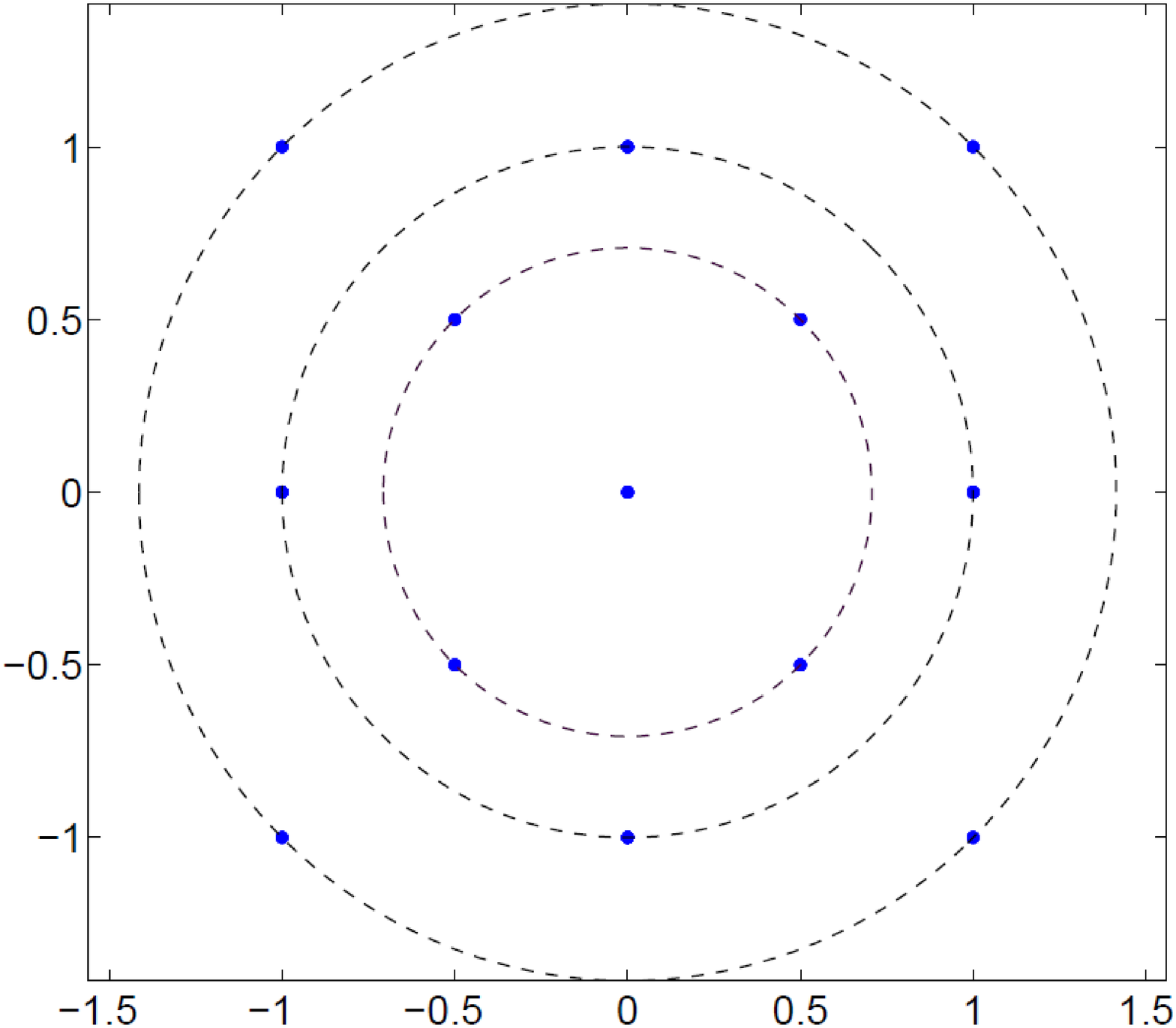}
\label{fig:4qam_sfs}
}
\caption{$4-$QAM constellation, its  difference constellation and singular fade states}
\label{fig:4QAM}
\end{figure}
%%%%%%%%%%%%%%%%%%%%%%%%%%%%%%%%%%%%%%%%%%%%%%%%
%%%%%%%%%%%%%%%%%%%%%%%%%%%%%%%%%%%%%%%%%%%%%%%%
\begin{definition}
\cite{LB} The Gaussian integers are the elements of the set $\mathbb{Z}[j]=\{a+bj : a,b \in \mathbb{Z}\}$ where $\mathbb{Z}$ denotes the set of integers. 
\end{definition}

The signal points in the difference constellation are Gaussian integers. To get the number of singular fade states for square QAM signal sets, the notion of primes and relatively primes in the set of Gaussian integers is useful.

\begin{definition}
\cite{LB} A Gaussian integer $\alpha$ is called a Gaussian prime if and only if the only Gaussian integers that divide $\alpha$ are: $1,-1,j,-j,\alpha,-\alpha, \alpha j$ and $-\alpha j.$ The Gaussian integers which are invertible in $\mathbb{Z}[j]$ are called units in $\mathbb{Z}[j]$ and they are $\pm 1$ and $\pm j.$  Let $\alpha, \beta \in \mathbb{Z}[j]$. If the only common divisors of $\alpha$ and $\beta$ are units, we say $\alpha$ and $\beta$ are relatively prime.
\end{definition}

\begin{lemma}
\label{no_sfs_qam}
The number of singular fade states for the square $M$-QAM signal set, denoted by
$N_{M-QAM}$ is given by $$N_{M-QAM} = 4+ 8 \phi(\Delta{S}^{+})$$ where $\phi(\Delta{S}^{+})$ is the number of relative prime pairs in $\Delta{S}^{+}.$
\end{lemma}

\begin{proof}
All the possible ratios of elements from $\Delta{S}$ give singular fade states. We consider only ratios in $\Delta{S}^+$ and multiply the number of possible such ratios with a factor of 4 to account the ratios with points in all the other quadrants. To avoid multiplicity while counting we take only relative prime pairs in $\Delta{S}^+$ and one such pair $(a,b)$ gives two singular fade states $a/b$ and $b/a$. Because of this the multiplication factor becomes 8. Finally, the factor 4 is added to count the units.
\end{proof}
%%%%%%%%%%%%%%%%%%%%%%%%%%%%%%%%%%%%%%%%%%%%%%%%%
\begin{example}
For a 4-QAM signal set shown in Fig.\ref{fig:4qam} the number of singular fade states, $N(z_{4-QAM})$ is given by $4+ 8.1=12$. Scaled $\Delta{S}^{+}$ have only two elements $\{1,1+j\}$ in this case as shown in Fig.\ref{fig:4qamdiffcons}. They form one relatively prime pair. The singular fade states $\pm1,\pm j,\pm1\pm j,\frac{1}{\pm1\pm j}$ is shown in Fig. \ref{fig:4qam_sfs}.
\end{example}
%%%%%%%%%%%%%%%%%%%%%%%%%%%%%%%%%%%%%%%%%%%%%%%%%
\begin{table}
\centering
 \caption{Prime factors of Gaussian integers in $\Delta{S}^+$}
 \label{prifactor}
\begin{tabular}{|c|c|c|}
\hline Elements in $\Delta{S}^{+}$ & Prime factors  & No. of relatively prime pairs\\ 
\hline 1 &  1 & 11\\ 
\hline 1+j& 1+j & 6\\
\hline 2& 1+j & 6\\
\hline 1+2j& 1+2j & 10\\
\hline 2+j& 2+j & 10\\
\hline 2+2j& 1+j & 6\\
\hline 3& 3 & 10\\
\hline 3+j& (1+j),(1+2j)& 5\\
\hline 1+3j& (1+j),(2+j)& 5\\
\hline 3+2j& 3+2j & 11\\
\hline 2+3j& 2+3j & 11\\
\hline 3+3j& (1+j),3 & 5\\
%\hline Total No. of relatively prime pairs for $\Delta{S}^+$ for 16-QAM & & 96 \\
\hline
\end{tabular}
\vspace{-.1 cm}
\end{table}
%%%%%%%%%%%%%%%% 

\begin{example}
Consider the case of 16-QAM signal set. Table \ref{prifactor} discusses the prime factorization of the elements in $\Delta{S}^{+}$. From the table there are 96 relatively prime pairs, but it counts the pair $(a,b)$ and $(b,a)$ separately. So there are 48 distinct pairs of relative primes, and from Lemma \ref{no_sfs_qam}, $N_{16-QAM}$ turns to be 388. The singular fade states of 16-QAM is shown in Fig.\ref{fig:16QAM_sing}.
\end{example}
%%%%%%%%%%%%%%%%%%%%%%%%%%%%%%%%%%%%%%%%%%%%%%%%%%%%%%%%%
\begin{table}[htbp]
\centering
 \caption{Comparison between $M$-PSK and $M$-QAM on number of singular fade states}
 \label{comp_psk_qam}
\begin{tabular}{|c|c|c|}
\hline $M$  & No. of singular fade & No. of singular fade \\ 
            & states for $M$-PSK &  states for $M$-QAM\\
\hline 4 &  12 & 12\\ 
\hline 16 & 912 & 388\\
\hline 64 & 63,552 & 8388\\
\hline
\end{tabular}
\vspace{-.1 cm}
\end{table}
%%%%%%%%%%%%%%%%%%%%%%%%%%%%%%%
%%%%%%%%%%%%%%%%%%%%%%%%%%%%%%%%%%%%%%%%%%%%%%%%%%%
\subsection{Singular fade states of $M$-PSK and $M$-QAM signal sets}
\label{subsec_3_2}
In this section we show that the number of singular fade states for $M$-QAM signal sets is lesser in comparison with that of $M$-PSK signal sets. The advantages of this are two fold- QAM offers better distance performance and it requires lesser number of overhead bits since the required number of relay clusterings are lesser in the case of QAM compared with PSK.
%%%%%%%%%%%%%%%%%%%%%%%%%%%%%%%%%%%%%%%%%%%%%%%%%%%%
\begin{lemma}
\label{upper_qam}
The number of singular fade states for $M$-QAM signal set is upper bounded by $4(n^2-n+1)$, where $n=\dfrac{[(2 \sqrt M -1)^2-1]}{4},$ which is same as $4(M^2)-(2M-1)\sqrt{M} +1).$ 
\end{lemma}
\begin{proof}
There are $[(2 \sqrt M -1)^2-1]$ non zero signal points in $\Delta{S}$ which are distributed equally in each quadrant, i.e., the number of signal points in  $\Delta{S}^+$, $\dfrac{[(2 \sqrt M -1)^2-1]}{4}$ which we denote by $n.$  %From Definition \ref{sfs_alter_def} it can be seen that the number of singular fade states is related to the number of different mappings possible with the $n$ signal points. 
The maximum number of relatively prime pairs in a set of $n$ Gaussian integers is $\frac{n(n-1)}{2}$. Since an upper bound is of interest we substitute this in Lemma \ref{no_sfs_qam} instead of $\phi(\Delta{S}^+).$ This completes the proof.
\end{proof}
%%%%%%%%%%%%%%%%%%%%%%%%%%%%%%%%%%%%%%%%%%%%%%%%%%%%%%%

The number of singular fade states for $M$-QAM signal set is lesser in comparison with that of $M$-PSK signal sets. In \cite{VNR} it is shown that the number of singular fade states for $M$-PSK signal set is $M(\frac{M^2}{4}-\frac{M}{2}+1)$, in $\mathcal{O}(M^3)$. From Lemma \ref{upper_qam}, an upper bound on the number of singular fade states for $M$-QAM is in $\mathcal{O}(M^2).$

%%%%%%%%%%%%%%%%%%%%%%%%%%%%%%%%%%%%%%%%%%%%%%%%%%%%%%%
\begin{example}
The singular fade states of 16-PSK signal set is given in Fig.\ref{fig:16psk_sing}. There are 912 singular fade states in total.
%\begin{figure}[h]
%\centering
%%\vspace{-.8 cm}
%\includegraphics[totalheight=2in,width=3.5in]{16psk_sing.eps}
%%\vspace{-2 cm}
%\caption{Singular fade states for 16-PSK signal sets.}     
%\label{fig:16psk_sing}        
%\end{figure}
\end{example}
%%%%%%%%%%%%%%%%%%%%%%%%%%%%%%%%%%%%%%%%%%%%%%%%%%%%%%%%%

The advantage of square QAM constellation is highly effective in higher order constellations, for example 64-QAM is having 8,388 singular fade states where as a 64-PSK has 63,552 singular fade states and relay has to adaptively use 63,552 clusterings. With the use of square QAM constellations the complexity is enormously reduced. 
%%%%%%%%%%%%%%%%%%%%%%%%%%%%%%%%%%%%%%%%%%%%%%%%%
\begin{figure*}[]
\centering
\hspace{2cm}
\subfigure[16-QAM ]{
\includegraphics[totalheight=2.3in,width=2.3in]{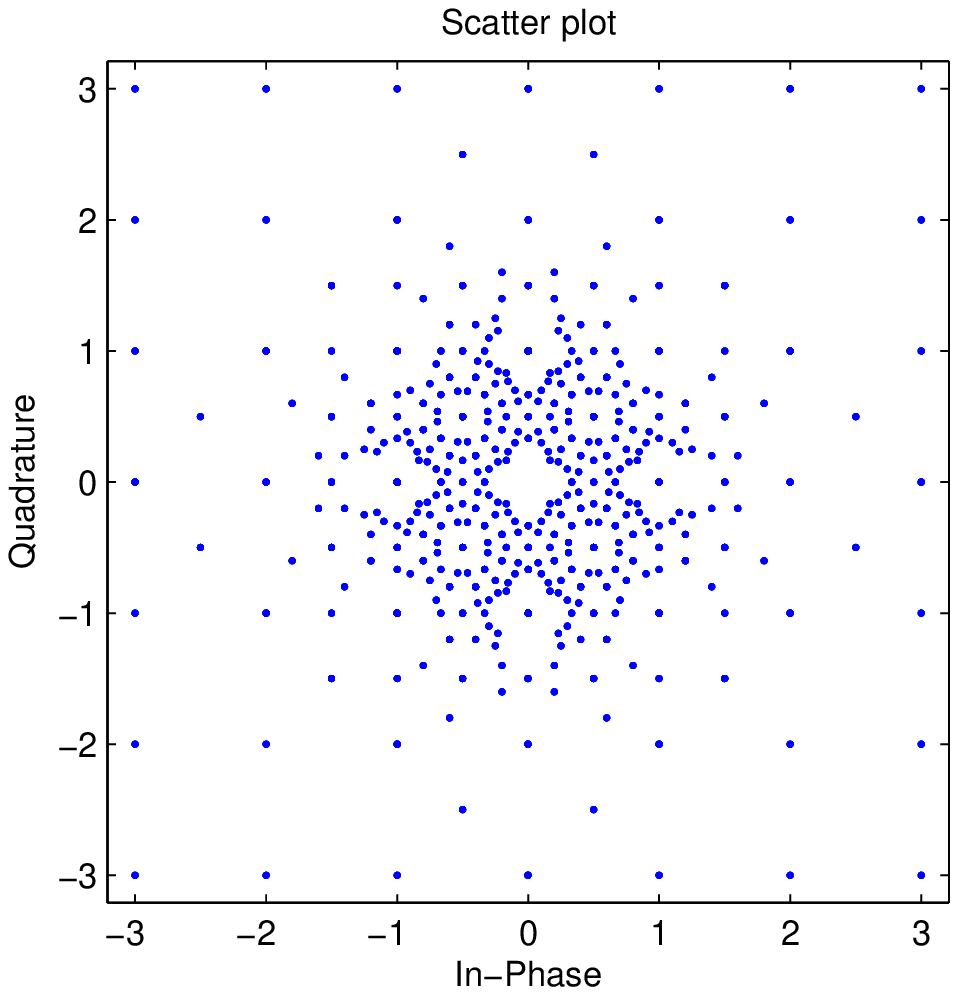}
\label{fig:16QAM_sing}
}
\qquad
\hspace{-1.5cm}
\subfigure[ 16-PSK ]{
\includegraphics[totalheight=2.25in,width=4in]{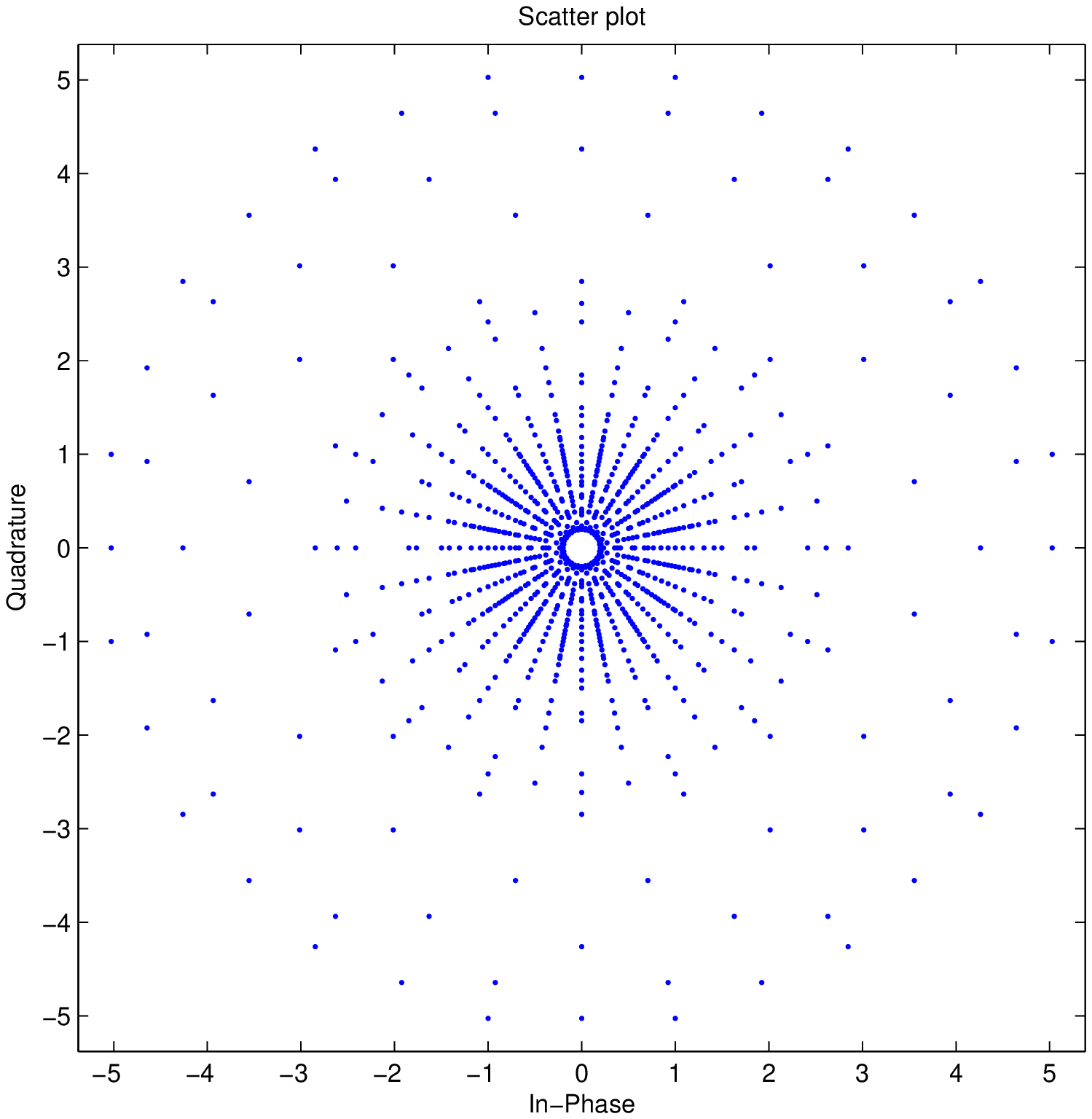}
\label{fig:16psk_sing}
}
\caption{Singular Fade States for $16-$QAM and $16-$PSK modulation schemes}
%\label{fig:16psk_sing}
\end{figure*}

%%%%%%%%%%%%%%%%%%%%%%%%%%%%%%%%%%%%%%%%%%%%%%%%%%
%%%%%%%%%%%%%%%%%%%%%%%%%%%%%%%%%%%
\section{Exclusive Law and Latin Squares}
\label{sec3}
\begin{definition} \cite{Rod} A Latin Square L of order $M$ with the symbols from the set $\mathbb{Z}_t=\{0,1, \cdots ,t-1\}$ is an \textit{M} $\times$ \textit{M}  array, in which each cell contains one symbol and each symbol occurs at most once in each row and column. 
\end{definition}

 In \cite{NVR} it is shown that when the end nodes use signal sets of same size all the relay clusterings which satisfy exclusive-law can be equivalently representable by Latin Squares, with the rows (columns) indexed by the constellation point used by node A (B) and the clusterings are obtained by taking all the slots in Latin Squares which are mapped to the same symbol in one cluster.
 %%%%%%%%%%%%%%%%%%%%%%%%%%%%%%%%%%%%%%%%%%%%%%%%%%%%%% 
 \subsection{Removing Singular fade states and Constrained Latin Squares}
The minimum size of the constellations needed in the  BC phase is $M,$ but it is observed that in some cases relay may not be able to remove the singular fade states with $t=M$ and $t > M$ results in severe performance degradation in the MA phase \cite{APT1}. Let $(k,l)(k^{\prime},l^{\prime})$ be the pairs which give same point in the effective constellation $\mathcal{S}_R$ at the relay for a singular fade state, where $k,k^{\prime},l,l^{\prime} \in \{0,1,....,M-1\}$ and $k,k^{\prime}$ are the constellation points used by node A and $l,l^{\prime}$ are the corresponding constellation points used by node B. If they are not clustered together, the minimum cluster distance will be zero. To avoid this, those pairs should be in same cluster. This requirement is termed as a {\it singularity-removal constraint}. So, we need to obtain Latin Squares which can remove singular fade states and with minimum value for $t.$ Towards this end, initially we fill the slots in the $\textit{M}\times\textit{M}$ array  such that for the slots corresponding to a singularity-removal constraint the same element is used to fill slots. This removes that particular singular fade state. Such a partially filled Latin Square is called a {\it Constrained Partially Latin Square} (CPLS). After this,  to make this a Latin Square, we try to fill the other slots of the CPLS with minimum number of symbols.
%%%%%%%%%%%%%%%%%%%%%%%%%%%%%%%%%%%%%%%%%%%%%%%%%%%%%%%%
%%%%%%%%%%%%%%%%%%%%%%%%%%%%%%%%%%%%%%%%%%%%
\begin{definition} A Latin Square $L^T$ is said to be the Transpose of a Latin Square $L$, if $L^T(i,j)=L(j,i)$ for all $i,j \in \{0,1,2,..,M-1\}.$
\end{definition}
%%%%%%%%%%%%%%%%%%%%%%%%%%%%%%%%%%%%%%%%%%%%%%%
\begin{lemma}
For any constellation, if the Latin Square $L$ removes the singular fade state $z$ then the  Latin Square $L^T$ will remove the singular fade state $z^{-1}.$
\end{lemma}
%%%%%%%%%%%%%%%%%%%%%%%%%%%%%%%%%%%%%
\begin{proof}
Let the singular fade state $z$ as given in \eqref{sing_expression} with constraint $\{(x_A,x_B),(x_A^{\prime},x_B^{\prime})\}.$ Then, by taking the inverse
\begin{align*}
z^{-1}=\dfrac{x_B^{\prime}-x_B}{x_A-x_A^{\prime}}.
\end{align*}
\noindent
Now the constraints are modified to $\{(x_B,x_A),(x_B^{\prime},x_A^{\prime})\},$ i.e., the role of node A and node B are interchanged, which clearly results in the transpose of the Latin Square.
\end{proof}
%%%%%%%%%%%%%%%%%%%%%%%%%%%%%%%%%%%%%%%%%%%

From the above lemma, it is clear that we have to get Latin Squares only for singular fade states $|z| \leq 1$ or $|z| \geq 1$. 

The square QAM signal set has a symmetry which is $\pi/2$ degrees of rotation. This results in a reduction of the number of required Latin Squares by a factor 4 as shown in the following lemma. 

\begin{lemma}
If $L$ is a Latin Square that removes a singular fade state $z$, then there exist a column permutation of $L$ such that the permuted Latin Square $L^\prime$ removes the singular fade state $z e^{j \pi/2}.$
\end{lemma}

\begin{proof}
For the singular fade state $z$ is given in \eqref{sing_expression} with constraint $\{(x_A,x_B),(x_A^{\prime},x_B^{\prime})\}$, the singular fade state $z e^{j \pi/2}$ is given by
\begin{align*}
z e^{j \pi/2}= \dfrac{[x_A-x_A^{\prime}]}{[x_B^{\prime}-x_B]}e^{j \pi/2}\\
\Longrightarrow z e^{j \pi/2}= \dfrac{[x_A -x_A^{\prime} ]}{[x_B^{\prime}e^{-j \pi/2}-x_B e^{-j \pi/2}]}.\\
\end{align*}
\noindent
Since in the square QAM constellation there exist signal points with $x_B^\prime e^{-j \pi/2}$ and $x_B e^{-j \pi/2}$, all the constraints are changed but the new constraints are obtainable from the permutation of signal points in the constellation used by node B. The columns of the Latin Squares are indexed by the signal points used by B and the effective permutation in the constellation is representable by column permutation in the Latin Square.
\end{proof}
 
%%%%%%%%%%%?%%%%%%%%%%%%%%%%%%%%%%%%%%%%%%%%%%%%%%%%%%%%%%%%%%%%%%%%%%
Note that the fade state $z=1$ or $(\gamma=1, \theta=0)$ is a singular fade state for any signal set. 
\begin{definition} A Latin Square which removes the singular fade state $z=1$ for a signal set is said to be a standard Latin Square for that signal set. 
\end{definition} 

When the signal sets is a $2^\lambda$-PSK signal set then, in \cite{VNR} it has been shown that the Latin Square obtained by Exclusive-OR (XOR) is a standard Latin Square for any integer $\lambda.$ It turns out that for $M-$QAM signal sets the  Latin Square given by bitwise Exclusive-Or (XOR) is not a standard Latin Square for any $M>4.$ This can be easily seen as follows: Any square $M$-QAM signal set ($M>4$) has points of the form $a, a+jb, a-jb$, for some integers $a$ and $b$. For $z=1$, the effective constellation at R during the MA phase contains the point $2a$ can result in at least two different ways, since $2a=a+za= (a+jb)+z(a-jb)$ for $z=1$. Let $ l_1, l_2$ and $l_3$ denote the labels for $a, a+jb,$ and $a-jb$ respectively. For the singular fade state $z=1,$ we have $\lbrace(l_1,l_1),(l_2,l_3) \rbrace$ as a singularity removal constraint. But the Latin Square obtained by bitwise XOR mapping does not satisfy this constraint since $l_1\oplus l_1=0 \neq l_2\oplus l_3$. 

\subsection{Standard Latin Square for $\sqrt M-$PAM}
In this subsection, we obtain standard Latin Squares for $\sqrt M-$PAM signal sets. 
%%% 
\begin{definition}
An $M \times M$ Latin square in which each row is obtained by a left cyclic shift of the previous row is called a left-cyclic Latin Square.
\end{definition}
%%%%%
\begin{lemma}
For a $\sqrt M$-PAM signal set a left-cyclic Latin Square removes the singular fade state $z=1$.  
\end{lemma}
%%%%%
\begin{proof}
Consider the $\sqrt M$-PAM signal set with the signal points labelled from left to right as discussed in Section \ref{sec2}. Let $\{(k_1,l_1)(k_2,l_2)\}$ be a singularity removal constraint. To get the same point in the received constellation at the relay R, when $z=1,$ we have $k_1+l_1=k_2+l_2.$ Consider the following two cases satisfying this equality.
\noindent
Case (i): $k_2=l_1, l_2=k_1$ In this case the constraint becomes $\{(k_1,l_1)(l_1,k_1)\},$ i.e., the Latin Square which removes $z=1$ should be symmetric about main diagonal.\\
\noindent
Case (ii): $k_2=k_1+m, l_2=l_1-m$ for any $m \leq \sqrt M.$, The constraint now becomes $\{(k_1,l_1)(k_1+m,l_1-m)\}$ which means the symbol in $k_1$-th row and $l_1$-th column should be repeated in the $k_1+1$-th row and the $l_1-1$-th column. \\
It is easily seen that a left-cyclic Latin Square satisfies both this requirements.
\end{proof}
%%%%%%%%%%
\begin{example}
Consider the received constellation at the relay when the end nodes use 4-PAM constellation and let the channel condition be $z=1$ as given in Fig.\ref{fig:pam_received}.
The singularity removal constraints are 
\begin{align*}
\{(0,1)(1,0)\},~\{(0,2)(1,1)(2,0)\},~\{(0,3)(1,2)(2,1)(3,0)\}, \\ \{(1,3)(2,2)(3,1)\},~ \mbox{and} \{(2,3)(3,2)\}.
\end{align*}
\noindent
The Latin Square which removes this singular fade state is given in Fig.\ref{cyclic_LS}.

\begin{figure}[h]
\centering
%\vspace{-.8 cm}
\includegraphics[totalheight=.3in,width=2.5in]{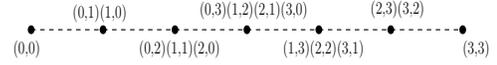}
%\vspace{-2 cm}
\caption{Received Constellation at the relay for $z=1$.}     
\label{fig:pam_received}        
\end{figure}
%%%%%%%%%%%%%%%%%%%%%%%%%%%%%%%%%%%%%%%%%%%%%%%%%%%%%%%%%
\begin{figure*}[htbp]
\centering
\begin{tabular}{|c|c|c|c||c|c|c|c||c|c|c|c||c|c|c|c|}
\hline 0 & 1 & 2 & 3 & 4 & 5 & 6 & 7 & 8 & 9 & 10 & 11 & 12 & 13 & 14 & 15\\
\hline 1 & 2 & 3 & 0 & 5 & 6 & 7 & 4 & 9 & 10 & 11 & 8 & 13 & 14 & 15 & 12\\ 
\hline 2 & 3 & 0 & 1 & 6 & 7 & 4 & 5 & 10 & 11 & 8 & 9 & 14 & 15 & 12 & 13\\
\hline 3 & 0 & 1 & 2 & 7 & 4 & 5 & 6 & 11 & 8 & 9 & 10 & 15 & 12 & 13 & 14\\
\hline
\hline 4 & 5 & 6 & 7 & 8 & 9 & 10 & 11 & 12 & 13 & 14 & 15 & 0 & 1 & 2 & 3\\
\hline 5 & 6 & 7 & 4 & 9 & 10 & 11 & 8 & 13 & 14 & 15 & 12 & 1 & 2 & 3 & 0\\ 
\hline 6 & 7 & 4 & 5 & 10 & 11 & 8 & 9 & 14 & 15 & 12 & 13 & 2 & 3 & 0 & 1\\
\hline 7 & 4 & 5 & 6 & 11 & 8 & 9 & 10 & 15 & 12 & 13 & 14 & 3 & 0 & 1 & 2\\
\hline
\hline 8 & 9 & 10 & 11 & 12 & 13 & 14 & 15 & 0 & 1 & 2 & 3 & 4 & 5 & 6 & 7\\
\hline 9 & 10 & 11 & 8 & 13 & 14 & 15 & 12 & 1 & 2 & 3 & 0 & 5 & 6 & 7 & 4\\ 
\hline 10 & 11 & 8 & 9 & 14 & 15 & 12 & 13 & 2 & 3 & 0 & 1 & 6 & 7 & 4 & 5\\
\hline 11 & 8 & 9 & 10 & 15 & 12 & 13 & 14 & 3 & 0 & 1 & 2 & 7 & 4 & 5 & 6\\
\hline
\hline 12 & 13 & 14 & 15 & 0 & 1 & 2 & 3 & 4 & 5 & 6 & 7 & 8 & 9 & 10 & 11\\
\hline 13 & 14 & 15 & 12 & 1 & 2 & 3 & 0 & 5 & 6 & 7 & 4 & 9 & 10 & 11 & 8\\ 
\hline 14 & 15 & 12 & 13 & 2 & 3 & 0 & 1 & 6 & 7 & 4 & 5 & 10 & 11 & 8 & 9\\
\hline 15 & 12 & 13 & 14 & 3 & 0 & 1 & 2 & 7 & 4 & 5 & 6 & 11 & 8 & 9 & 10\\
\hline 
\end{tabular}
\caption{Standard Latin Square $L_{QAM}$ for 16-QAM.}
\vspace{-.1 cm}
\label{qam_lat}
\end{figure*}
%%%%%%%%%%%%%%%%%%%%%%%%%%%%%%%%%%%%%%%%%%%%%%%%%%%%%%%%%%%%%%
\begin{figure}[h]
\centering
\begin{tabular}{|c|c|c|c|}
\hline 0 & 1 & 2 & 3 \\ 
\hline 1 & 2 & 3 & 0 \\ 
\hline 2 & 3 & 0 & 1\\
\hline 3 & 0 & 1 & 2 \\
\hline 
\end{tabular}
\caption{Left-cyclic Latin Square to remove the singular fade state $z=1$}
\vspace{-.1 cm}
\label{cyclic_LS}
\end{figure}
\end{example}

%%%%%%%%%%%%%%%%%%%%%%%%%%%%%%%%%%%%%%%%%%%%%%%%%%%%%%%%%%%%%%
\subsection{Standard Latin Square for $M-$QAM}
In this subsection standard Latin Square for a square $M$-QAM constellation is obtained from that of $\sqrt M$-PAM constellation. 

Let $PAM-i,$ for $i=1,2,\cdots, {\sqrt M},$ denote the symbol set consisting of $\sqrt M$ symbols $\{(i-1)\sqrt M, ((i-1)\sqrt M)+1, ((i-1)\sqrt M)+2,\cdots,((i-1)\sqrt M)+(\sqrt M-1)\}.$ Let $L_{PAM-i}$ denote the standard Latin Square with  symbol set $PAM-i$ for $\sqrt M$-PAM and also let  $L_{QAM}$ denote the standard Latin Square for $M$-QAM. Then, $L_{QAM}$ is given in terms of $L_{PAM-i},$ $i=1,2,\cdots, \sqrt M,$ as the block left-cyclic Latin Square shown in Fig. \ref{fig:qam_const}. This is formally shown in the following Lemma.
%%%%%%%%%%%%%%%%%%%%%%%%

\begin{lemma}
\label{qam_ls}
Let $PAM-i$  for $i=1,2,\cdots, {\sqrt M},$ denote the symbol set consisting  of $\sqrt M$ symbols $\{(i-1)\sqrt M, ((i-1)\sqrt M)+1, ((i-1)\sqrt M)+2,\cdots,((i-1)\sqrt M)+(\sqrt M-1)\}$ and let  $L_{PAM-i}$ stand for the Latin Square that removes the singular fade state $z=1$ with a symbol set $PAM-i$ for $\sqrt M$-PAM. Then arranging the cyclic Latin Squares $L_{PAM-i}$ as shown Fig.\ref{fig:qam_const} where each row is a blockwise left-cyclically shifted version of the previous row results in a Latin Square which removes the singular fade state $z=1$ for $M$-QAM. 
\end{lemma}
%%%%%%%%%%%%%%%%%%%
\begin{proof}
Note that the matrix in Fig. \ref{fig:qam_const} is a $M \times M$  matrix, which is also a $\sqrt M \times \sqrt M$ block left-cyclic  matrix where each block is a $\sqrt M \times \sqrt M$  left-cyclic matrix $L_{PAM-i}$ for some $i.$      

Let $a_{1}+jb_{1}, a_{2}+jb_{2},a_{1}^{\prime}+jb_{1}^{\prime}$ and $a_{2}^{\prime}+jb_{2}^{\prime}$, where $a_{i},a_{i}^{\prime},b_{i}$ and $b_{i}^{\prime} \in \{-(\sqrt M-1),-(\sqrt M-3),\cdots,(\sqrt M-3),(\sqrt M-1)\}$ for $i \in \{1,2\}$ be four $M$-QAM constellation points such that $a_{1}+jb_{1}$ and $a_{1}^{\prime}+jb_{1}^{\prime}$ are used by node A and $a_{2}+jb_{2}$ and $a_{2}^{\prime}+jb_{2}^{\prime}$ are used by end node B, and result in a same point in the effective received constellation at the relay node for singular fade state $z=1$, i.e.,
\begin{align*}
a_{1}+jb_{1}+a_{2}+jb_{2}=a_{1}^{\prime}+jb_{1}^{\prime}+a_{2}^{\prime}+jb_{2}^{\prime}.
\end{align*}

Let $a_{1}^{\prime}=a_1+m_1$ and $b_{1}^{\prime}=b_1+m_2$ where $m_1,m_2 \in \{-2(\sqrt M -1), -2(\sqrt M -2),\cdots,2(\sqrt M -2),2(\sqrt M -1)\}$. Then, $a_{2}^{\prime}=a_2-m_1$ and $b_{2}^{\prime}=b_2-m_2.$ Then, using the map defined in \ref{mumap}, let
%It can be seen that for a $a_1+jb_1$ and $a_2+jb_2$, we need to consider values of $m_1$ and $m_2$ only when $a_1+m_1+j(b_1+m_2)$ and $a_2-m_1+j(b_2-m_2)$ are $M$-QAM constellation points, such $m_1, m_2$ values are called valid. Let $\mu$ be tdenote the complex number corresponding to a M-QAM constellation point to symbol(from $\mathbb{Z}_M$) mapping. Let
\begin{align*}
k_1=\mu(a_{1}+jb_{1})\\ 
l_1=\mu(a_{2}+jb_{2})
\end{align*}
\vspace{-.7cm}
\begin{align*}
k_2=\mu(a_{1}^{\prime}+jb_{1}^{\prime})=\mu(a_1+m_1+j(b_1+m_2))\\
l_2=\mu(a_{2}^{\prime}+jb_{2}^{\prime})=\mu(a_2-m_1+j(b_2-m_2))
\end{align*}
Since, for $z=1,$ the four complex numbers result in the same point in the effective constellation at the relay, $\{(k_1,l_1)(k_2,l_2)\}$ is a singularity removal constraint for $z=1.$ From the 
above equations it follows that
%mapping $\mu$ defined in section \ref{sfs_of_qam} $\mu(A_{mI}+jA_{mQ})=\frac{1}{2}[(\sqrt M -1 +A_{mI})\sqrt M + (\sqrt M -1 +A_{mQ})]$, it is obtained that for valid $m_1$ and $m_2$,
\begin{align*}
k_2=k_1+\frac{1}{2} (m_1 \sqrt M+m_2)\\
l_2=l_1-\frac{1}{2}(m_1 \sqrt M+m_2)
\end{align*}
The above equations precisely mean the construction shown in Fig.\ref{fig:qam_const}. This completes the proof. 

%Now we observe valid $m_1$ and $m_2$ values, for $i \in \mathbb{Z}_{\sqrt M}$ when $i\sqrt M \leq k_1 < (i+1)\sqrt M$ valid values for $m_1$ are $-2i \leq m_1 \leq 2(\sqrt M-(i+1))$ and when $i\sqrt M \leq l_1 < (i+1)\sqrt M$ valid values for $m_1$ are $-2(\sqrt M-(i+1)) \leq m_1 \leq 2i$. Similarly when $k_1$ mod $\sqrt M=i$ valid values for $m_2$ are $-2i \leq m_2 \leq 2(\sqrt M-(i+1))$ and when $l_1$ mod $\sqrt M=i$ valid values for $m_2$ are $-2(\sqrt M-(i+1)) \leq m_2 \leq 2i$. For valid $m_1,m_2$ values the singularity removal constraint $\{(k_1,l_1)(k_1+2m_1+\frac{m_2}{2},l_1-(2m_1+\frac{m_2}{2}))\}$ says that the symbol in $k_1^{th}$ row and $l_1^{th}$ column should be repeated in the $(k_1+2m_1+\frac{m_2}{2})^{th}$ row and $(l_1-(2m_1+\frac{m_2}{2}))^{th}$ column.
%This results in a cyclic shift towards left of previous row in certain columns given by the values of $m_1$ and $m_2$. By completing the $M \times M$ Latin Square from this gives the construction shown in Fig.\ref{fig:qam_const}.
\end{proof}

%%%%%%%%%%%%%
\begin{figure}[h]
\centering
%\vspace{-.8 cm}
\includegraphics[totalheight=2.8in,width=2.8in]{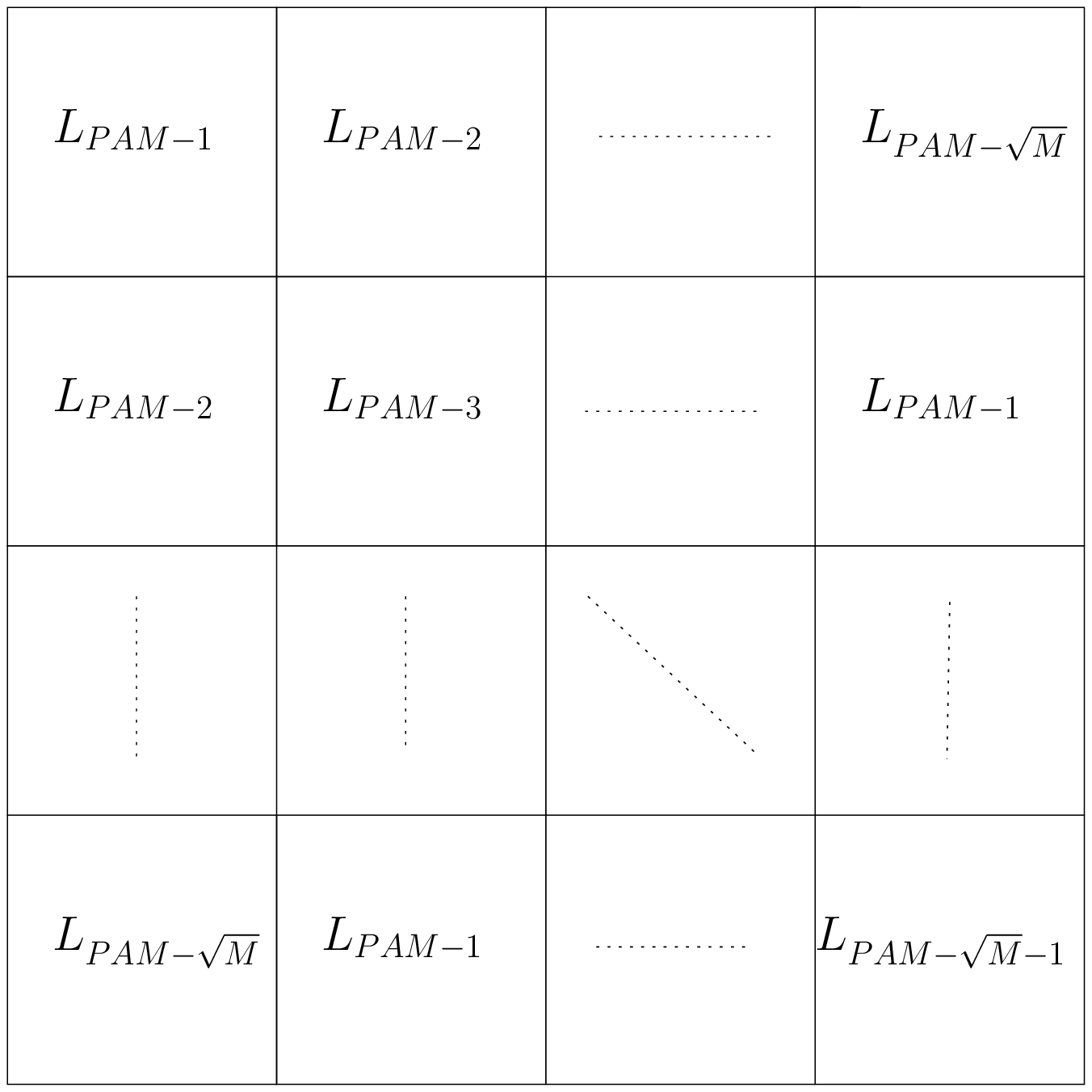}
%\vspace{-2 cm}
\caption{Construction of $L_{QAM}$ for $z=1$.}     
\label{fig:qam_const}        
\end{figure}
%%%%%%%
The standard Latin Square for 16-QAM is shown in Fig.\ref{qam_lat}.

%%%%%%%%%%%%%%%%%%%%%%%%%%%%%%%%%%%%%%%%%%%%%%%%%%%%%%%%%%%%%%%%%%%%%%%
%%% To be added later after getting corrected with Vishnu %%%%%%%%%%%%%%
%%%%%%%%%%% BEGINS %%%%%%%%%%%%%%%%%%%%%%%%%%%%%%%%%
We define a minimal Latin Square as,
\begin{definition} An $M \times M$ Latin Square with $M$ symbols is termed as a minimal Latin Square.
\end{definition}
\section{SIMULATION RESULTS}
\label{sec4}
 \begin{figure}[t]
\centering
%\vspace{-.8 cm}
\includegraphics[totalheight=2in,width=3in]{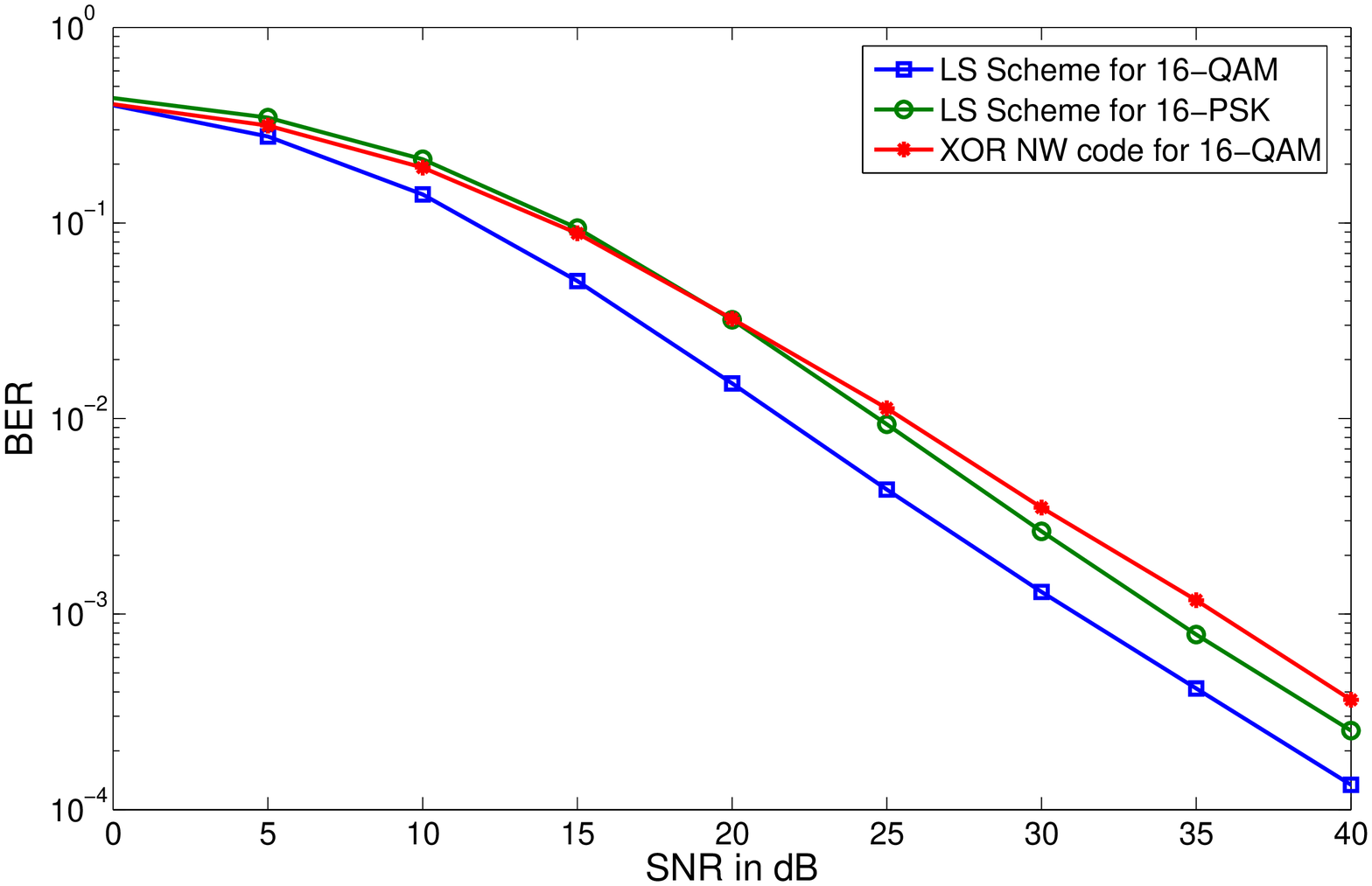}
%\vspace{-2 cm}
\caption{SNR vs BER for different schemes when the end nodes use 16-QAM and 16-PSK for a Rayleigh fading scenario.}     
\label{fig:rayleigh}        
\end{figure}

The proposed Latin Square (LS) Scheme (\cite{NVR}) is based on removing the singular fade states. For 16-PSK all the 912 singular fade states can be removed with minimal Latin Squares, but for 16-QAM some singular fade states cannot be removed with minimal Latin Squares. Since 16-QAM have only 388 singular fade states, in comparison with 912 singular fade states of 16-PSK, 16-QAM offers better distance distribution in the MA stage. For a given average energy, the end to end BER is a function of distance distribution of the constellations used at the end nodes as well as at the relay. The simulation results for the end to end BER  as a function of SNR is presented in this section for different fading scenarios.

Consider the case when $H_A , H_B , H_{A}^{\prime}$ and $H_{B}^{\prime}$ are distributed according to Rayleigh distribution, with the variances of all the fading links are assumed to be 0 dB. The end to end BER as a function of SNR in dB when the end nodes use 16-QAM signal sets as well as 16-PSK signal sets with same average energy is given in Fig.\ref{fig:rayleigh}. The end to end BER for XOR network code for 16-QAM is also given. It can be observed that the LS Scheme for 16-QAM outperforms LS Scheme for 16-PSK as well as XOR network code. 

Consider the case when $H_A , H_B , H_{A}^{\prime}$ and $H_B^{\prime}$ are distributed according to Rician distribution, with the Rician factor of 5 dB and the variances of all the fading links are assumed to be 0 dB. In Fig.\ref{fig:rician} the end to end BER as a function of SNR in dB for LS scheme for 16-PSK, 16-QAM and XOR network coding for 16-QAM is given. It is observed that the LS scheme gives large gain over the XOR network coding scheme. The LS scheme for QAM is  better in end to end BER performance in comparison with the LS scheme for PSK.

%%%%%%%%%%%%%%%%%%%%%%%%%%%%%%%%%%%%%%%%%%%%%%%%%%%%%%%%%%%%%%%%%%%%%%%%%%%%%%%%%%%%%%%%%%%%
\section{DISCUSSION}
In this paper, for the design of modulation schemes for the physical layer network-coded two way relaying scenario with the protocol which employs two phases: Multiple access (MA) Phase and Broadcast (BC) phase, with both end nodes use square QAM constellation is studied. We showed that there are many advantages of using square QAM constellation. With the help of the relation between exclusive law satisfying clusterings and Latin Squares we propose a method to remove the singular fade states.  This relation is  used to get all the maps to be used at the relay efficiently. We proposed a construction scheme to get the Latin Square for square QAM constellation from PAM constellation. Here we concentrated only on singular fade states and the clusterings to remove that with only the minimum cluster distance under consideration. We are not considering the entire distance profile as done in \cite{APT1}. Unlike in the case of \cite{APT1}, we could remove most of the singular fade states with standard Latin Square and its isotopes. We presented the simulation results showing the end to end BER performance when the end nodes use PSK constellation as well as QAM constellations. 
%%%%%%%%%%%%%%%%%%%%%%%%%%%%%%%%%%%%%%%%%%%%%%%%%%%%%%%%%%%%%%%%%%%%%%%%%%%%%%%%%%%%%%%%%%%%%%%%%%%%%%%%%%%%%%%%%55  
\begin{figure}[t]
\centering
%\vspace{-.8 cm}
\includegraphics[totalheight=2in,width=3in]{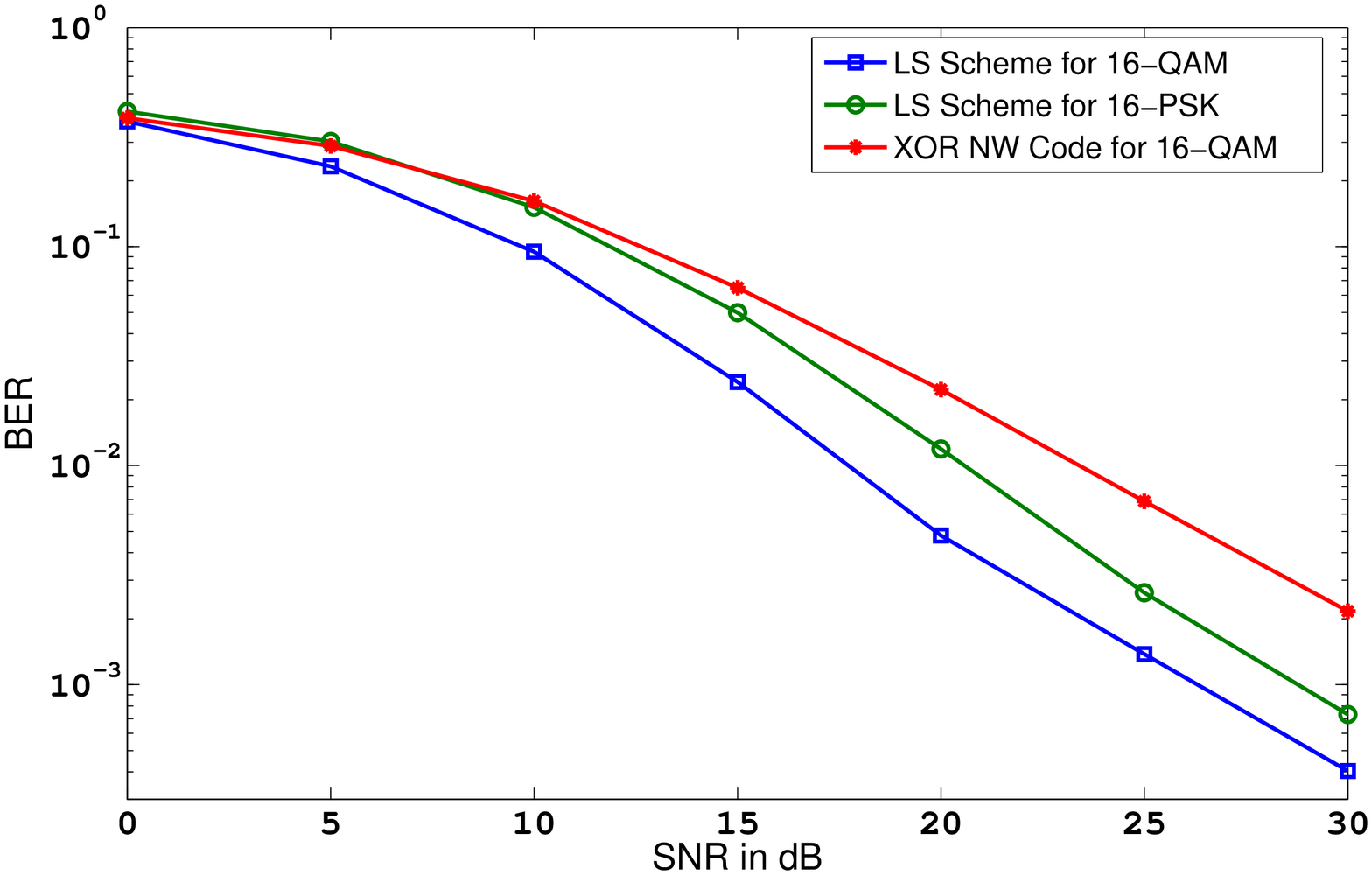}
%\vspace{-2 cm}
\caption{SNR vs BER for different schemes when the end nodes use 16-QAM and 16-PSK for a Rician fading scenario with Rician factor 5 dB. }     
\label{fig:rician}        
\end{figure}
%%%%%%%%%%%%%%%%%%%%%%%%%%%%%%%%%%%%%%%%%%%%%%%%%%%%%%%%%%%%%%%%%%%%%%%%%%%%%%%%%%%%%%%%%%
%\section*{Acknowledgement}
%This work was supported  partly by the DRDO-IISc program on Advanced Research in Mathematical Engineering through a research grant as well as the INAE Chair Professorship grant to B.~S.~Rajan.
%%%%%%%%%%%%%%%%%%%%%%%%%%%%%%%%%%%%%%%%%%%%%%%%%%%%%%%%%%%%%%%%%%%%%%%%%%%%%%%%%%%%%%

\end{document}